\documentclass[a4paper,11pt,notitlepage]{article}

\usepackage[top=2cm, bottom=2cm, left=2cm, right=2cm]{geometry}
\usepackage{amsfonts}
\usepackage{amssymb}
\usepackage{graphicx}
\usepackage{amsmath}
\usepackage{graphicx}
\usepackage{pdfpages}
\usepackage{setspace}
\usepackage{multirow}
\usepackage[hidelinks]{hyperref}

\newtheorem{theorem}{Theorem}

\newtheorem{lemma}[theorem]{Lemma}

\newtheorem{proposition}[theorem]{Proposition}

\newenvironment{proof}[1][Proof]{\textbf{#1.} }{\ \rule{0.5em}{0.5em}}

\makeatletter
\renewcommand{\fnum@figure}[1]{\textbf{\figurename~\thefigure}.}
\makeatother
\makeatletter
\renewcommand{\fnum@table}[1]{\textbf{\tablename~\thetable}.}
\makeatother

\begin{document}
\singlespacing

\title{Circulant Singular Spectrum Analysis: A new automated procedure for signal extraction
\thanks{
Financial
support from the Spanish government, contract grants MINECO/FEDER
ECO2015-70331-C2-1-R, ECO2015-66593-P and ECO2016-76818-C3-3-P is
acknowledged.~}
}

\author{Juan B\'{o}galo \\
{\small Universidad de Alcal\'{a}}\\
{\small SPAIN} \and Pilar Poncela \\
{\small Universidad Aut\'{o}noma de Madrid}\\
{\small SPAIN} \and Eva Senra \\
{\small Universidad de Alcal\'{a}}\\
{\small SPAIN}}
\maketitle
\date{}

\begin{abstract}
\onehalfspacing
Sometimes, it is of interest to single out the fluctuations associated to a given frequency. We propose a new variant of SSA, Circulant SSA (CiSSA), that allows to extract the signal associated to any frequency 
specified beforehand. This is a novelty when compared with other procedures that need to identify ex-post the frequencies associated to extracted signals. We prove that CiSSA is asymptotically equivalent to 
these alternative procedures although with the advantage of avoiding the need of the subsequent frequency identification. We check its good performance and compare it to alternative SSA methods
through several simulations for linear and nonlinear time series. We also prove its validity in the nonstationary case. To show how it works with real data, we apply CiSSA to extract the business cycle and 
deseasonalize the Industrial Production Index of six countries. Economists follow this indicator in order to assess the state of the economy in real time. We find that the estimated cycles match the dated recessions from the OECD showing its reliability for business cycle analysis. Finally, we analyze the 
strong separability of the estimated components. In particular, we check that the deseasonalized time series do not show any evidence of residual seasonality.

\bigskip



\end{abstract}
\newpage

\doublespacing
\section{Introduction}

Singular Spectrum Analysis (SSA) is a non-parametric procedure
based on subspace algorithms for signal extraction \cite{Golyandina_Zhigljavsky_2013}. The main task in SSA is to extract the underlying signals of a time series like the trend, cycle,
seasonal and irregular components. It has been applied to a wide range of
time series problems, besides signal processing \cite{Golyandina_2020}, like forecasting \cite{Khan_Poskitt_2017}, missing value imputation \cite{Mahmoudvand_Rodrigues_2016} or functional time series \cite{Haghbin_others_2019} among others.
SSA builds a trajectory matrix by putting
together lagged pieces of the original time series and works with the
Singular Value Decomposition of this matrix. It can be viewed as applying Principal Component (PC) analysis to the columns of the trajectory matrix.

SSA has been applied in different disciplines as several authors illustrate (see \cite{Golyandina_Korobeynikov_2014} and the references therein). For instance, there are recent applications in biometry \cite{Safi_others_2018}, climatology \cite{Yurova_others_2019}, energy \cite{Kumar_Jain_2010} or volcanic activity \cite{Bozzo_others_2010}.
In business and economics, SSA has been reviewed for economic and financial time series, focusing on forecasting and business cycle analysis \cite{Hassani_Thomakos_2010}. On the other hand, analyze the effects of forecasting with SSA before and after the 2008 recession \cite{Hassani_others_2013a, Silva_Hassani_2015}, forecast the inflation dynamics \cite{Hassani_others_2013b} and forecast industrial production with multivariate SSA \cite{Hassani_others_2013c} are examples of empirical applications on forecasting with SSA. Related to the business cycle, the corresponding to US is tracked \cite{Carvalho_others_2012}, the real time
nowcasting of the output gap is studied \cite{Carvalho_Rua_2017} and the economic cycles and their synchronization in three European countries are analized \cite{Sella_others_2016}. SSA has also been applied to estimate stochastic volatility models \cite{Arteche_Garcia_2017} and intraday data forecasting \cite{Lahmiri_2018}.

The common practice when applying SSA is to extract the Principal Components of the trajectory matrix and to identify afterwards the frequencies associated to the extracted components, by analysing their estimated periodogram \cite{Carvalho_Rua_2017, Alexandrov_Golyandina_2005, Vautard_others_1992} or response function \cite{Kume_2013, Tome_others_2018} just to cite a few. 
Though there are fast computing algorithms for the eigenvalues and eigenvectors of Toeplitz matrices \cite{Elden_Sjostrom_1996, Korobeynikov_2009}, the use of circulant matrices have a great advantage as their eigenvalues and eigenvectors have a closed form. In a different context \cite{Das_others_2018} also use circulant matrices within the MUSIC algorithm but restricted to signals that are approximately periodic and deterministic. 

We propose a new SSA methodology (CiSSA), that can be applied to any time series signal, based on circulant matrices that, once the user has decided beforehand the frequency of interest, it automatically matches this frequency with specific principal components. 
Circulant matrices become relevant in this
setup, as their eigenstructure can be obtained as a function of the frequency and, therefore, we
can automatically identify their eigenvalues and eigenvectors associated to any particular frequency.
Our approach, CiSSA, valid in a general setting, automatically identifies the eigenvalues and eigenvectors associated to any
particular frequency using circulant matrices. Moreover, we obtain an easy way to evaluate the spectral density since the later approximates with the eigenvalues at the matched frequencies. 

CiSSA seems to perform and compare well with previous versions of SSA, like Basic or Toeplitz SSA, despite introducing its automatization. In order to show this, first, we have proved that CiSSA is asymptotically equivalent to these alternative procedures. Second, we have checked its performance in practice through several sets of simulations for linear and nonlinear models. 
Finally, we have extended its validity for non-stationary time series. Although SSA has been successfully used in non-stationary time series previously, e.g., \cite{Ma_others_2010}, our value added is that we apply it in an automated way and also provide a theoretical background overcoming the assumption of stationarity.

In summary, our contribution is to propose a new version of SSA, Circulant
SSA, for signal extraction in an automated way valid for any type of signal. With this new version, we make heavy use of circulant matrices and obtain reliable components associated to any pre-specified frequency, both for stationary and non-stationary time series. 

We illustrate CiSSA procedure by applying it to the Industrial Production
Index (IP) of six developed countries. IP is a relevant indicator to follow
the business cycle and its seasonally adjusted signal is followed in real
time to monitor the economy. We check that our estimated cycles match the official dating of recessions provided by the OECD. Finally, we also study the strong separability of
the estimated components.

The structure of this paper is as follows: Section 2 briefly describes the
SSA technique. Section 3 proposes our new SSA procedure, named after
Circulant SSA, proves its asymptotic equivalence to Basic and Toeplitz SSA and extends its use for nonstationary time series. Section
4 presents a set of simulations to check the properties of the proposed
methodology. Section 5 applies it to the estimation of the business cycle
of the industrial production index in six countries and checks the good
properties of the obtained estimations as compared with the official dating of the OECD. Finally, Section 6 concludes.

\section{SSA methodology}

The origin of SSA dates back to 1986 with the publication of the papers by
Broomhead and King \cite{Broomhead_King_1986a, Broomhead_King_1986b} and Fraedrich \cite{Fraedrich_1986}. In 1989, Vautard and Ghil \cite{Vautard_Ghil_1989} introduce Toeplitz SSA for stationary time series and, three years later, Vautard et al. \cite{Vautard_others_1992} derive the algorithm called diagonal averaging to obtain the extracted components with the length of the original series. At the same time, and independently, the so-called Caterpillar technique was developed in the former Soviet Union \cite{Danilov_Zhigljavsky_1997}.

In this section we briefly describe the steps used in SSA to decompose a
time series in its unobserved components (trend, cycle,...). Basically, SSA
is a technique in two stages: decomposition and reconstruction. In the first
stage, decomposition, we transform the original vector of data into a
related trajectory matrix and perform its singular value decomposition to
obtain the so called elementary matrices. This corresponds to steps 1 and 2
in the algorithm. In the second stage, reconstruction, (steps 3 and 4 of the
algorithm) we classify the elementary matrices into disjoint groups associating each
group to an unobserved component (trend, cycle,...). Finally, we transform
every group into an unobserved component of the same size of the original
time series by diagonal averaging. 

To proceed with the algorithm, let $\left\{ x_{t}\right\}$ denote a stochastic process $t \in \cal T$ and let $\left\{ x_{t}\right\}_{t=1} ^{T}$ be a realization \footnote{For simplicity, we use the same notation for the stochastic process and for the observed time series. It will be clear from the context if we are referring to the population or to the sample. If it were not, we would explicetely clarify it in the main text.} of $x_t$ of length $T$, $\mathbf{x}=(x_{1},...,x_{T})^{\prime}$, where the prime denotes transpose and $L$ a positive integer, called the
window length, such that $1<L<T/2$. The Basic SSA or Broomhead-King (BK)
procedure involves the following 4 steps:

\textbf{1st step: Embedding}

From the original time series we will obtain an $L\times N$ trajectory
matrix $\mathbf{X}$, $N=T-L+1$, as follows
\begin{equation}
\mathbf{X=}\left( \mathbf{x}_{1}|...|\mathbf{x}_{N}\right) =\left( 
\begin{array}{ccccc}
x_{1} & x_{2} & x_{3} & ... & x_{N} \\ 
x_{2} & x_{3} & x_{4} & ... & x_{N+1} \\ 
\vdots & \vdots & \vdots & \vdots & \vdots \\ 
x_{L} & x_{L+1} & x_{L+2} & ... & x_{T}%
\end{array}%
\right) \label{trajectory}
\end{equation}%
where $\mathbf{x}_{j}=(x_{j},...,x_{j+L-1})^{\prime }$ indicates the $L \times 1$ vector
with origin at time $j$. Notice that the trajectory matrix $%
\mathbf{X}$ is Hankel and both, by columns and rows, we obtain subseries of
the original one.

\textbf{2nd step: Decomposition}

In this step, we perform the singular value decomposition (SVD) of the
trajectory matrix $\mathbf{X=UD}^{1/2}\mathbf{V}^{\prime }$ where $\ \mathbf{%
U}$ is the $L\times L$ matrix whose columns $\mathbf{u}_{k}$ are the $%
L\times 1$ eigenvectors of the second moment matrix $\mathbf{S=XX}^{\prime }$%
, $\mathbf{D}=diag(\tau _{1},...,\tau _{L})$, $\tau _{1}\geq ...\geq \tau
_{L}\geq 0$, are the eigenvalues of $\mathbf{S}$ and $\mathbf{V}$ is the $%
N\times L$ matrix whose $L$ columns $\mathbf{v}_{k}$ are the $N\times 1$ 
eigenvectors of $\mathbf{X}^{\prime }\mathbf{X}$ associated to nonzero
eigenvalues. This decomposition allows to write $\mathbf{X}$ as the sum of
the so-called elementary matrices $\mathbf{X}_{k}$ of rank 1,%
\begin{equation*}
\mathbf{X=}\sum_{k=1}^{r}\mathbf{X}_{k}=\sum_{k=1}^{r}\mathbf{u}_{k}\mathbf{w%
}_{k}^{\prime },
\end{equation*}%
where $\mathbf{w}_{k}=\mathbf{X}^{\prime }\mathbf{u}_{k}=\sqrt{\tau _{k}}%
\mathbf{v}_{k}$, being $\sqrt{\tau _{k}}$ the singular values of the $\mathbf{%
X}$ matrix, and $r=\max_{\tau _{k}>0}\{k\}$=rank($\mathbf{X}$).

\textbf{3rd step: Grouping}

Under the assumption of weak separability given in \cite{Golyandina_others_2001},
we group the elementary matrices $\mathbf{X}_{k}$ into $G$ disjoint groups
summing up the matrices within each group. Let $I_{j}, j=1,...,G$ be each disjoint group of indexes associated to the
corresponding eigenvectors. The matrix $\mathbf{X}_{I_{j}}=\sum_{k\in I_j} \mathbf{X}%
_k$ is associated to the $I_{j}$ group. \ The
decomposition of the trajectory matrix into these groups is given by $%
\mathbf{X}=\mathbf{X}_{I_{1}}+...+\mathbf{X}_{I_{G}}.$ The contribution of
the component coming from matrix $\mathbf{X}_{I_{j}}$ is given by $%
\sum_{k\in _{I_{j}}}\tau _{k}/\sum_{k=1}^{r}\tau _{k}.$

\textbf{4th step: Reconstruction}

Let $\mathbf{X}_{I_{j}}=(\widetilde{x}_{ij})$. In this step, each matrix $%
\mathbf{X}_{I_{j}}$ is transformed into a new time series of the same
length $T$ as the original one, denoted as $\widetilde{\mathbf{x}}^{(j)}=(%
\widetilde{x}_{1}^{(j)},...,\widetilde{x}_{T}^{(j)})^{\prime }$ by diagonal
averaging. This is equivalent to averaging the elements of $\mathbf{X}_{I_{j}}$ over its antidiagonals as
follows%
\begin{equation*}
\widetilde{x}_{t}^{(j)}=\left\{ 
\begin{array}{l}
\frac{1}{t}\sum_{i=1}^{t}\widetilde{x}_{i,t-i+1},\qquad 1\leq t<L \\ 
\frac{1}{L}\sum_{i=1}^{L}\widetilde{x}_{i,t-i+1},\qquad L\leq t\leq N \\ 
\frac{1}{T-t+1}\sum_{i=L-N+1}^{T-N+1}\widetilde{x}_{i,t-i+1},\qquad N<t\leq T%
\end{array}%
\right .
\end{equation*}

\bigskip

The alternative Toeplitz SSA or Vautard-Ghil (VG) relies on the assumption
that $\mathbf{x}$ is stationary and zero mean and it performs the orthogonal diagonalization in step 2 from an alternative matrix $\mathbf{S}_{T}\mathbf{=(}s_{ij}\mathbf{)}$ where
\begin{equation}
s_{ij}\mathbf{=}\frac{1}{T-|i-j|}\sum_{m=1}^{T-|i-j|}x_{m}x_{m+|i-j|},\hspace{1cm}1\leq i,j\leq L\:. \label{vg}
\end{equation}
In this case, the matrix $\mathbf{S}_{T}$ is the sample lagged
variance-covariance matrix of the original series, a symmetric Toeplitz
matrix. The set $(\tau _{k},\mathbf{u}_{k},\mathbf{w}_{k})$ is named the $k$%
-th eigentriple. The rest of the algorithm remains unchanged.

\section{Circulant SSA}

SSA in any of its variants requires to identify the harmonic frequencies of the extracted components and this makes necessary the analysis of the periodogram. To try to automate SSA, several strategies have been proposed such as find the correlations at different lags between the elements of two eigenvectors, associated to almost identical eigenvalues to test if they are in quadrature \cite{Ghil_Mo_1991}; effect a test based on the periodogram to establish if a pair of eigenvectors are associated to the same harmonic \cite{Vautard_others_1992}; introduce optimal thresholds for grouping eigenvectors linked to nearby frequencies in order to assign them to the same harmonic \cite{Alexandrov_Golyandina_2005, Alexandrov_Golyandina_2004}; perform a spectral-based Fisher $g$ test to asses certain principal components to the business cycle frequency \cite{Carvalho_Rua_2017}; considering eigenvectors as filters \cite{Kume_2013} group the outputs according to their frequency reponse \cite{Tome_others_2018}; and even apply cluster techniques for grouping the elementary components based on k-means \cite{Alonso_Salgado_2008} or hierarchical clustering \cite{Bilancia_Campobasso_2010}. Nevertheless, whatever procedure is used, the grouping of frequencies is made after the elementary components are extracted. Since the pairs of eigenvalues and eigenvectors are obtained, not as a function of the frequency, but rather on a decreasing magnitude, this means that the grouping is done with uncertainty. A partial solution is provided by linking the eigenvalues-eigenvectors as a function of the frequency for symmetric positive definite Toeplitz matrices \cite{Bozzo_others_2010}. However, the analytic form of the eigenvalues for this type of matrices is only known for heptadiagonal matrices \cite{Solary_2013}. 
We generalize the link between the eigenstructure of a matrix and the associated frequencies by the use of circulant matrices allowing non-periodic signals.

In this section, we propose an automated version of SSA based on circulant
matrices. First, we deal with the stationary case and, later on, we will
extend our proposal to the nonstationary case.

\subsection{Stationary case}

In this subsection we propose to apply SSA to an alternative matrix of
second moments that is circulant. In this case, we have closed solutions form
eigenvalues-eigenvectors that are linked to the desirable specific
frequencies. We show the asymptotic equivalence between the traditional
Toeplitz matrices used in SSA and our proposed circulant matrices.
Based on all the previous results we propose a new alghorithm that we name
Circulant SSA (CiSSA).

Toeplitz matrices appear when considering the population second order
moments of the trajectory matrix. Let $\{x_{t}\}$ be an infinite, zero mean
stationary time series whose autocovariances are given by $\gamma
_{m}=E(x_{t}x_{t-m})$, $m=0,1,...$ and its spectral density function, a real continuous and $2\pi $-periodic function, denoted by $f.$ Let

\begin{equation}
\mathbf{\Gamma }_{L}(f)=\left( 
\begin{array}{ccccc}
\gamma _{0} & \gamma _{1} & \gamma _{2} & ... & \gamma _{L-1} \\ 
\gamma _{1} & \gamma _{0} & \gamma _{1} & ... & \gamma _{L-2} \\ 
\vdots & \vdots & \vdots & \vdots & \vdots \\ 
\gamma _{L-1} & \gamma _{L-2} & \gamma _{L-3} & ... & \gamma _{0}%
\end{array}%
\right) \label{Gamma_L}
\end{equation}%
be the $L\times L$ matrix that collects these second moments. Notice that $%
\mathbf{\Gamma }_{L}(f)$ is a symmetric Toeplitz matrix that depends on the
spectral density $f$ through the covariances $\gamma _{m}$. Recall that \ $%
\gamma _{m}=\int_{0}^{1}f(w)\exp (i2\pi mw)dw$ for any integer $m$ where $w \in \left[ 0,\ 1 \right]$
is the frequency in cycles per unit of time.

Analytic expressions for the eigenvalues of Toeplitz matrices are only known up to heptadiagonal matrices. To be able to have closed solutions of the eigenvalues and eigenvectors for any dimension, we use a special case of Toeplitz matrices that are the circulant ones. In a circulant matrix every row is a right cyclic shift of the row above as follows:

\begin{equation*}
\mathbf{C}_{L}(f)=\left( 
\begin{array}{ccccc}
c_{0} & c_{1} & c_{2} & ... & c_{L-1} \\ 
c_{L-1} & c_{0} & c_{1} & ... & c_{L-2} \\ 
\vdots & \vdots & \vdots & \vdots & \vdots \\ 
c_{1} & c_{2} & c_{3} & ... & c_{0}%
\end{array}%
\right) .
\end{equation*}

The eigenvalues and eigenvectors of a circulant matrix have a closed form \cite{Lancaster_1969}. The $k$-th eigenvalue of the $L\times L$ circulant matrix $\mathbf{C}_{L}(f)$ is given by
\begin{equation*}
\lambda _{L,k}=\sum_{m=0}^{L-1}c_{m}\exp \left( i2\pi m\frac{k-1}{L}\right)
\end{equation*}
for $k=1,...,L$ and its associated eigenvector can be written as
\begin{equation}
\mathbf{u}_{k}=L^{-1/2}(u_{k,1,}...,u_{k,L})^{\prime } \label{u_k}
\end{equation}
where $u_{k,j}=\exp \left( -i2\pi (j-1)\frac{k-1}{L}\right) $.

In particular, if we consider the circulant matrix of order $L\times L$ with
elements $c_{m}$ defined as:
\begin{equation}
c_{m}=\frac{1}{L}\sum_{j=0}^{L-1}f\left( \frac{j}{L}\right) \exp \left(
i2\pi m\frac{j}{L}\right) ,\qquad m=0,1,...,L-1,
\end{equation}%
we have two interesting results \cite{Gray_1972}. First, the eigenvalues of this
circulant matrix coincide with the spectral density evaluated at points $w_k=\frac{k-1}{L}$, 
\begin{equation}
\lambda _{L,k}=f\left( \frac{k-1}{L}\right) . \label{auto}
\end{equation}%
And, second, the matrices $\mathbf{\Gamma }_{L}(f)$ and $\mathbf{C}_{L}(f)$
are asymptotically equivalent as $L\rightarrow \infty $, $\mathbf{\Gamma }%
_{L}(f)$ $\sim $ $\mathbf{C}_{L}(f),$ in the sense that both matrices have
bounded eigenvalues \cite{Tilli_1998} and $\underset{L\rightarrow \infty }{\lim }\frac{%
\left\Vert \mathbf{\Gamma }_{L}(f)-\mathbf{C}_{L}(f)\right\Vert _{F}}{\sqrt{L%
}}=0$, where $\left\Vert \text{\textperiodcentered }\right\Vert _{F}$ is the
Frobenius norm. Moreover, the eigenvalues of both
matrices $\mathbf{\Gamma }_{L}(f)$ and $\mathbf{C}_{L}(f)$ are
asymptotically equally distributed in the sense of Weyl\footnote{%
Two sets of bounded real numbers $\left\{ a_{n,k}\right\} _{k=1}^{n}$ and $%
\left\{ b_{n,k}\right\} _{k=1}^{n}$are asymptotically equally distributed in
the sense of Weyl if\textit{\ }for a given continuous function $F$ on the
interval $\left[ -K,K\right] $, it holds that $\underset{n\rightarrow \infty 
}{\lim }\frac{\underset{k=1}{\overset{n}{\sum }}\left(
F(a_{n,k})-F(b_{n,k})\right) }{n}=0.$} as a consequence of the fundamental
theorem of Szeg\"{o} \cite[p. 64]{Grenander_Szego_1958} as it is shown in \cite{Trench_2003}.

To obtain a more operational version of the procedure, we consider the circulant matrix ${\mathbf{C}}_{L}(\widetilde{f})$ whose elements $\widetilde{c}_{m}$ are given by \cite{Pearl_1973}:

\begin{equation}
\widetilde{c}_{m}=\frac{L-m}{L}\gamma _{m}+\frac{m}{L}\gamma _{L-m},\quad
m=0,1,...,L-1\:, \label{ctilde_m}
\end{equation}
where the generating function $\widetilde{f}$ is an approximation of the spectral density $f$. Besides that, \cite{Pearl_1973} shows that $\mathbf{\Gamma }_{L}(f)$ is
asymptotically equivalent to ${\mathbf{C}}_{L}(\widetilde{f}).$ By the
transitivity property, the three matrices $\mathbf{\Gamma }_{L}(f)$, $
{\mathbf{C}}_{L}(\widetilde{f})$ and $\mathbf{C}_{L}(f)$ are asymptotically
equivalent.

Therefore, our proposal will consist on using the eigenstructure of a
circulant matrix ${\mathbf{C}}_{L}(\widetilde{f})$ with elements given by (\ref
{ctilde_m}) and, by (\ref{auto}) associate the $k^{th}$ eigenvalue and corresponding eigenvector to the
frequency $w_k=\frac{k-1}{L}$. Moreover, again by (\ref{auto}) the spectral
density is easily evaluated at frequencies $w_k$ by the eigenvalues of the matrix $
{\mathbf{C}}_{L}(\widetilde{f})$.

Finally, going to the sample we have to work with estimated, rather than
population, quantities. So, we substitute the population autocovariances $%
\left\{ \gamma _{m}\right\} _{m=0}^{L-1}$, by the sample second moments $%
\left\{\widehat{\gamma }_{m}\right\} _{m=0}^{L-1}$ where $\widehat{\gamma }_{m},m=0,...,L-1$ is defined as 
\begin{equation*}
\widehat{\gamma }_{m}=\frac{1}{T-m}\sum_{t=1}^{T-m}x_{t}x_{t+m}\:.
\end{equation*}%
Since the sample autocovariances converge in probability to the population
autocovariances, we define $\mathbf{S}_{C}$ with elements given by 
\begin{equation}
\widehat{c}_{m}=\frac{L-m}{L}\widehat{\gamma }_{m}+\frac{m}{L}\widehat{\gamma }_{L-m},\qquad m=0,1,...,L-1\:. \label{c_m}
\end{equation}

In what follows, we describe our new proposed algorithm, named Circulant
SSA. Given the time series data $\left\{ x_{t}\right\} _{t=1}^{T}$:

\textbf{1st step: Embedding}. This step is as before.

\textbf{2nd step: Decomposition. }Compute the circulant matrix $\mathbf{S}_{C}$ whose
elements are given in (\ref{c_m}). Find the eigenvalues $\widehat{\lambda }%
_{k}$ of $\mathbf{S}_{C}$ and based on (\ref{auto}), associate the $k$-th eigenvalue and corresponding eigenvector
to the frequency $w_k=\frac{k-1}{L},k=1,...,L.$

\textbf{3rd step: Grouping. } Given the symmetry of the spectral density,
we have that $\widehat{\lambda }_{k}$ $=$ $\widehat{\lambda }_{L+2-k}$.
Their corresponding eigenvectors given by (\ref{u_k}) are complex,
therefore, they are conjugated complex by pairs, $\mathbf{u}_{k}=\overline{%
\mathbf{u}}_{L+2-k}$ where $\overline{\mathbf{v}}$ indicates the complex
conjugate of a vector $\mathbf{v}$, and $\mathbf{u}_{k}^{\prime }%
\mathbf{X}$ and $\mathbf{u}_{L+2-k}^{\prime }\mathbf{X}$\
correspond to the same harmonic period. We proceed as follows to transform
them in pairs of real eigenvectors in order to compute the associated
components.

To form the elementary matrices we first form the groups of 2 elements $%
B_{k}=\{k,L+2-k\}$ for $k=2,...,M$ with $B_{1}=\{1\}$ and $B_{\frac{L}{2}%
+1}=\left\{ \frac{L}{2}+1\right\} $ if $L$ is even. Second, we compute the
elementary matrix by frequency $\mathbf{X}_{B_{k}}$ as the sum of the two
elementary matrices $\mathbf{X}_{k}$ and $\mathbf{X}_{L+2-k}$, associated to
eigenvalues $\widehat{\lambda }_{k}$ and $\widehat{\lambda }_{L+2-k}$ \ and
frequency $w_k=\frac{k-1}{L}$,%
\begin{eqnarray*}
\mathbf{X}_{B_{k}} &=&\mathbf{X}_{k}+\mathbf{X}_{L+2-k} \\
&=&\mathbf{u}_{k}\overline{\mathbf{u}}_{k}^{\prime }\mathbf{X+u}_{L+2-k}%
\overline{\mathbf{u}}_{L+2-k}^{\prime }\mathbf{X} \\
&=&(\mathbf{u}_{k}\overline{\mathbf{u}}_{k}^{\prime }+\overline{\mathbf{u}}%
_{k}\mathbf{u}_{k}^{\prime })\mathbf{X} \\
&=&2(R_{\mathbf{u}_{k}}R_{\mathbf{u}_{k}}^{^{\prime }}+I_{\mathbf{u}_{k}}I_{%
\mathbf{u}_{k}}^{^{\prime }})\mathbf{X}
\end{eqnarray*}%
where $R_{\mathbf{u}_{k}}\,\ $denotes\ the real
part of $\mathbf{u}_{k}$ and $I_{\mathbf{u}_{k}}\,\ $its imaginary part.
Notice that the matrices $\mathbf{X}_{B_{k}},k=1,...,L,$ are real. 

\textbf{4th step: Reconstruction}. As before.

Notice that the elementary reconstructed series by frequency can be
automatically assigned to a component according to the goal of our analysis.

\subsection{Asymptotic equivalence of Basic, Toeplitz and Circulant SSA}
Toeplitz and Circulant SSA are modifications of the original Basic SSA. In this section, we will prove that the three vresions of SSA (Basic, Toeplitz and Circulant) are asymptotically equivalent according to the definition given in \cite{Gray_1972}. Later on, we will run some simulations to compare the performance of the three versions in finite samples.
\begin{theorem}
Given the $L \times N$ trajectory matrix $\mathbf{X}$ defined in (\ref{trajectory}), let $\mathbf{S}_B=\mathbf{XX'}/N$, $\mathbf{S}_{T}$ the Toeplitz matrix with elements defined by (\ref{vg}) and $\mathbf{S}_C$ the circulant matrix with elements given in (\ref{c_m}). Consider the sequence of matrices $\left\{ \mathbf{S}_B\right\} ,\left\{\mathbf{S}_{T}\right\} $ and $\left\{ \mathbf{
S}_{C}\right\} $ as $L \longrightarrow \infty$. Then $\mathbf{S}_{B}\sim
\mathbf{S}_{T}\sim \mathbf{S}_{C}$.
\end{theorem}

\begin{proof}
The proof is given in the appendix
\end{proof}

This theorem gives the basis to understand the similar results obtained in practice between Basic and Toeplitz SSA when the window length is very large (the larger, the better as the result is asymptotically). This was empirically shown using stationary time series in climate and geophisics \cite{Allen_Smith_1996, Ghil_others_2002}. Here, we provide a theoretical basis for these empirical findings. Additionally, we also extend the result for the new version of SSA that we have introduced in this paper, CiSSA.

\subsection{Nonstationary case}

In economics, many time series are nonstationary. That is to say, the
spectral density function has discontinuities. This has important
consequences in our analysis and we have to show that Circulant SSA can be
applied to nonstationary time series. The next theorem, a generalization
of the analogous Gray's theorem \cite[Theorem 3]{Gray_1974}, provides the theoretical background
needed to apply CiSSA to nonstationary time series.

\begin{theorem}
Let $\mathbf{T}_{L}(s)$ be a sequence of Toeplitz matrices with $s(w)$ a
real, continuous and 2$\pi $-periodic, such that $s(w)\geq 0$, where the
equality is reached in a finite number of points $H=\{w_{i}^{0},i=1,...,l\}$%
. Given a finite $\delta $, consider the disjoint sets%
\begin{equation*}
\Omega _{i}=\left\{ w\in \left[ w_{i}^{0}-a_{i},w_{i}^{0}+a_{i}\right]
|s(w)\leq \frac{1}{\delta }\right\} ,a_{i}\in 
\mathbb{R}^{+},\:i=1,...,l
\end{equation*}%
and let $g(w)$ be a function defined as
\begin{equation*}
g\left( w \right)=\left\{
\begin{array}{*{35}{l}}
 f\left( w \right)=\tfrac{1}{s\left( w \right)} & \text{if }w\notin \bigcup\nolimits_{i=1}^{l}{{{\Omega }_{i}}} \\
{{h}_{i}}\left( w \right) & \text{if }w\in {{\Omega }_{i}} \\
\end{array}\right.
\end{equation*}
where $h_{i}(w)$ is any real valued bounded function continuous in $\Omega
_{i}$ and symmetric around $w_{i}^{0}$. Let $M_{h_{i}}=$ sup $h_{i}<\infty $ and $m_{h_{i}}=$ inf $%
h_{i}=h_{i}\left( w_{i}^{0}-a_{i}\right) =h_{i}\left( w_{i}^{0}+a_{i}\right)
=\delta .$

Let $\rho _{L,k},k=1,...,L,$ be the eigenvalues of $\left( \mathbf{T}%
_{L}(s)\right) ^{-1}$ sorted in decreasing order and let $F(x)$ be a
continuous function in $\left[ \frac{1}{M_{s}},\max_{i}M_{h_{i}}\right] $
with $M_{s}=$ sup $s$, then%
\begin{equation}
\lim_{L\rightarrow \infty }\frac{1}{L}\sum_{k=1}^{L}F(\min (\rho _{L,k},\max
(\widetilde{g}_{k},\delta )))=\int\limits_{0}^{1}F(g(w))dw, \label{teorema}
\end{equation}%
where $\widetilde{g}_{k}$ are the values of $g(\frac{k-1}{L})$ sorted in
descending order.
\end{theorem}

\textbf{Proof}: The proof is given in the Appendix.

\bigskip

In a similar way to \cite{Gray_1974}, the theorem states that the sequence of
eigenvalues of the sequence of matrices $\left( \mathbf{T}_{L}(s)\right)
^{-1}$ are asymptotically equally distributed (in the sense of Weyl) as the
eigenvalues of the sequence of matrices $\mathbf{T}_{L}(g)$ up to a finite
value $\delta $ as $L$ tends to infinity. Moreover, the matrices $\mathbf{T}
_{L}(g)\sim \mathbf{C}_{L}(g)$ and, by Szeg\"{o}'s theorem, the eigenvalues
of the sequence of matrices $\mathbf{T}_{L}(g)$ are asymptotically equally
distributed as the eigenvalues of the sequence of matrices $\mathbf{C}
_{L}(g) $ up to a finite value $\delta $ as $L$ tends to infinity.

As a result, for a nonstationary series, the union of the estimation of the
pseudo-spectral density in a point of discontinuity with the estimations in the
adjoint frequencies through segments is an easy way of building the
functions $h_{i}.$ If all the functions $h_{i}$ are constant and equal to a
particular value $\delta $ finite, we have the particular case proved in
\cite[Theorem 3]{Gray_1974}. Therefore, the generalization to functions $h_{i}$ allows a
better approximation of the pseudo-spectral density when we increase the window
length.

\section{Simulations}

In this section we check the performance of our new proposal, Circulant SSA,
in finite samples and compare it with the competing SSA algorithms, i.e.
Basic SSA and Toeplitz SSA for a linear as well as a nonlinear time series
model. 
Even though SSA is nonparametric and, therefore, model free in this section we generate time series following a known model and check the basic statistical properties related to the signal extraction procedure. In particular, we check if the extracted signals are unbiased. 
These simulations generalize previous exercises \cite{Golyandina_2019} by including CiSSA, but also using more complex time series models in a linear and non-linear framework. 

\subsection{Linear time series}

The first model is a basic structural time series model%
\begin{equation}
x_{t}=T_{t}+c_{t}+s_{t}+e_{t} \label{gen}
\end{equation}%
where $T_{t}$ is the trend component, $c_{t}$ is the cycle, $s_{t}$ is the
seasonal component and $e_{t}$ is the irregular component. We assume an
integrated random walk for the trend \cite{Young_1984} given by 
\begin{eqnarray}
T_{t} &=&T_{t-1}+\beta _{t-1} \label{trend} \\
\beta _{t} &=&\beta _{t-1}+\eta _{t} \notag
\end{eqnarray}%
with $\eta _{t}\sim N(0,\sigma _{\eta }^{2}).$ The cyclical and seasonal
components are specified according to \cite{Durbin_Koopman_2012}, where the
cycle is given by the first component of the bivariate VAR(1) 
\begin{equation}
\left( 
\begin{array}{c}
c_{t} \\ 
\widetilde{c}_{t}%
\end{array}%
\right) =\rho _{c}\left( 
\begin{array}{cc}
\cos (2\pi w_{c}) & \sin (2\pi w_{c}) \\ 
-\sin (2\pi w_{c}) & \cos (2\pi w_{c})%
\end{array}%
\right) \left( 
\begin{array}{c}
c_{t-1} \\ 
\widetilde{c}_{t-1}%
\end{array}%
\right) +\left( 
\begin{array}{c}
\varepsilon _{t} \\ 
\widetilde{\varepsilon }_{t}%
\end{array}%
\right)  \label{cycle}
\end{equation}%
with $\left( 
\begin{array}{c}
\varepsilon _{t} \\ 
\widetilde{\varepsilon }_{t}%
\end{array}%
\right) \sim N(\mathbf{0},\sigma _{\varepsilon }^{2}I)$ and $\frac{1}{w_{c}}$
the period, $w_{c}\in \lbrack 0,1]$. And, the seasonal component is given by%
\begin{equation}
s_{t}=\sum_{j=1}^{[s/2]}a_{j,t}\cos (2\pi w_{j}t)+b_{j,t}\cos (2\pi w_{j}t)
\label{seasonal}
\end{equation}%
with $w_{j}=\frac{j}{s},j=1,...,[s/2]$ and $s$ the seasonal period, where $[$%
\textperiodcentered $]$ is the integer part and $a_{j,t}$ and $b_{j,t}$ are
two independent random walks with noise variances equal to $\sigma _{j}^{2}.$
Finally, the irregular component is white noise with variance $\sigma
_{e}^{2}.$ All the components are independent of each other. We set $\rho
_{c}=1,$ so the trend, cycle and seasonal components have a unit root. We
consider that the series are monthly with $s=12$ and cyclical period equal
to $\frac{1}{w_{c}}=48$ months. The sample size is $T=193$ and the noise
variances of the different components are given by $\sigma _{\eta
}^{2}=0.0006^{2}$, $\sigma _{j}^{2}=0.004^{2}$, $\sigma _{\varepsilon
}^{2}=0.008^{2}$ and $\sigma _{e}^{2}=0.06^{2}.$ We choose as window length $L=48$ because this value of $L$ is multiple of the seasonal period, it is equal to the cyclical period and $T-1$ is
multiple of $L$ \cite{Golyandina_Zhigljavsky_2013}.

The trend is related to frequency 0, the cycle to frequency 1/48 and the
seasonal components to frequencies 1/12, 1/6, 1/3, 1/4, 5/12 and 1/2. Given (%
\ref{auto}), we can recover the signal associated to a frequency $w=\frac{k-1%
}{L}$ by using the elementary components associated to eigenvalues $k$
and $k^{^{\prime }}=L+2-k,$ the latter by the symmetry of the spectral
density. Therefore, the trend is reconstructed with eigentriple 1, the
cyclical component with eigentriples\emph{\ }2 and 48, and the seasonal
components with eigentriples 5, 9, 13, 17, 21, 25, 29, 33, 37, 41 and 45.
For example, for the frequency $w=\frac{1}{12}$, we have that $\frac{k-1}{L}=%
\frac{1}{12}$, and therefore, we sum the elementary components $k=\frac{48}{%
12}+1=5$ and $k^{^{\prime }}=L+2-k=48+2-5=45$.

If the procedure for signal extraction works well, the simulated component $%
y_{t}$ ($y_{t}~$\ can be the trend, cycle or seasonal component) could be
written as%
\begin{equation*}
y_{t}=\widehat{y}_{t}+u_{t}
\end{equation*}%
where $u_{t}$ is the noise and $\widehat{y}_{t}$ is the extracted signal.
Then, in the regression 
\begin{equation}
y_{t}=a+b\widehat{y}_{t}+u_{t} \label{reg}
\end{equation}%
$a=0$ (unbiasedness) and $b=1\,$\ (the scale is not changed). Notice that $%
y_{t}$ and $\widehat{y}_{t}$ should be cointegrated. We simulate 10000 times
the model and perform signal extraction with Circulant SSA. Table \ref{Tabla_1} shows
the percentiles of the empirical distribution of the estimated coefficients
of the regression in (\ref{reg}).

\begin{table}[p]
\begin{center}
\begin{spacing}{1.5}
\begin{tabular}{c | c | c c c c c}
\hline
\multirow{2}{*}{\textbf{Statistic}} & \multirow{2}{*}{\textbf{Component}} & \multicolumn{5}{c}{\textbf{Quantiles}} \\
& & 5 & 25 & 50 & 75 & 95 \\ \hline
\multicolumn{7}{c}{\textbf{Circulant SSA}} \\ \hline
\multirow{3}{*}{$\hat{a}$} & Trend & -0.0613 & -0.0209 & -0.0006 & 0.0194 & 0.0600 \\
& Cycle & -0.0109 & -0.0043 & 0.0000 & 0.0045 & 0.0108 \\
& Seasonal & -0.0015 & -0.0006 & 0.0000 & 0.0006 & 0.0015 \\ \hline
\multirow{3}{*}{$\hat{b}$} & Trend & 0.9748 & 0.9951 & 1.0032 & 1.0143 & 1.0651 \\
& Cycle & 0.8481 & 0.9569 & 1.0029 & 1.0476 & 1.1340 \\
& Seasonal & 0.9451 & 0.9819 & 1.0049 & 1.0277 & 1.0630 \\ \hline
\multicolumn{7}{c}{\textbf{Basic SSA}} \\ \hline
\multirow{3}{*}{$\hat{a}$} & Trend & -0.0610 & -0.0206 & -0.0006 & 0.0191 & 0.0598 \\
& Cycle & -0.0165 & -0.0066 & 0.0001 & 0.0065 & 0.0167 \\
& Seasonal & -0.0033 & -0.0010 & 0.0000 & 0.0010 & 0.0033 \\ \hline
\multirow{3}{*}{$\hat{b}$} & Trend & 0.9881 & 1.0063 & 1.0153 & 1.0326 & 1.1292 \\
& Cycle & 0.7891 & 0.9618 & 1.0177 & 1.0794 & 1.2793 \\
& Seasonal & 0.9471 & 0.9911 & 1.0166 & 1.0431 & 1.0867 \\ \hline
\multicolumn{7}{c}{\textbf{Toeplitz SSA}} \\ \hline
\multirow{3}{*}{$\hat{a}$} & Trend & -0.0588 & -0.0203 & -0.0007 & 0.0186 & 0.0566 \\
& Cycle & -0.0178 & -0.0061 & 0.0001 & 0.0062 & 0.0170 \\
& Seasonal & -0.0017 & -0.0007 & 0.0000 & 0.0007 & 0.0018 \\ \hline
\multirow{3}{*}{$\hat{b}$} & Trend & 0.9820 & 1.0003 & 1.0088 & 1.0264 & 1.1415 \\
& Cycle & 0.7852 & 0.9863 & 1.0537 & 1.1310 & 1.2754 \\
& Seasonal & 0.9554 & 0.9982 & 1.0273 & 1.0605 & 1.1207 \\ \hline
\end{tabular}
\end{spacing}
\caption{Statistics related to the goodness of fit of the extracted signals for the different methods. Simulations for the linear model, N=10000. Columns show the quantiles of the empirical distribution of the estimated coefficients of the regression of the generated components over the estimated ones.}
\label{Tabla_1}
\end{center}
\end{table}

Table \ref{Tabla_1} shows that the median of the estimated intercept is almost zero for
the three estimated components (cycle, seasonal component and trend). The
median for the scale parameter $b$ is almost one for the three components,
but looking at the values for different quantiles, the empirical
distribution for the estimated $b$ associated to the cycle indicates a
larger dispersion.

The estimated residuals from equation (\ref{gen}) are given by $\widehat{e}%
_{t}=x_{t}-\widehat{T}_{t}-\widehat{c}_{t}-\widehat{s}_{t\text{ }}$, and
should be white noise, where $\widehat{T}_{t},\widehat{c}_{t},$ and $%
\widehat{s}_{t\text{ }}$are the estimates of the trend, cycle and seasonal
component respectively. In order to check this, we fit an AR(1) to $\widehat{%
e}_{t}$. Table \ref{Tabla_2} shows the quantiles of the empirical distributions of the
mean, standard error and autoregressive coefficient of the residuals of
the 10000 replications. The median of the mean and autoregressive
coefficient are close to zero. The median of the standard deviation is
0.0529 (the value used for the simulations was 0.06).

\begin{table}[h]
\begin{center}
\begin{spacing}{1.5}
\begin{tabular}{c | c c c c c}
\hline
\multirow{2}{*}{\textbf{Statistic}} & \multicolumn{5}{c}{\textbf{Quantiles}} \\
& 5 & 25 & 50 & 75 & 95 \\ \hline
Average & -0.0033 & -0.0012 & 0.0000 & 0.0011 & 0.0033 \\
Standard deviation & 0.0478 & 0.0508 & 0.0529 & 0.0551 & 0.0581 \\
AR(1) coefficient & -0.1693 & -0.0870 & -0.0313 & 0.0285 & 0.1075 \\ \hline
\end{tabular}
\end{spacing}
\caption{Statistics related to the residual term $\widehat{e}_{t}$ in Circulant SSA: Average, standard deviation and autoregressive coefficient of AR(1). Simulations for the linear model, N=10000.}
\label{Tabla_2}
\end{center}
\end{table}

The results from the simulations seem very good. In order to compare Circulant SSA with alternative algorithms as Basic and Toeplitz SSA we also simulate the linear model given by (\ref{gen}) and extract the trend, cycle and seasonal components for 10000 simulations. Basic and Toeplitz SSA require first to calculate the principal components and then to identify the frequency they represent with some procedure as stated in the first paragraph of this section. However, given that we are using simulated time series and we know beforehand the frequencies that might be more informative, we proceed in an heuristic way. According to model (\ref{gen}), we know that the informative frequencies are $\Omega =\left\{0,1/48,1/12,1/6,1/4,1/3,5/12,1/2\right\} $ and the window length $L=48$ coincides with the cycle periodicity and is multiple of the seasonal periodicity of a monthly time series. Also each eigenvector generates a linear subspace associated to a frequency. In this way, we calculate the periodogram for each eigenvector and obtain the frequency associated with the maximum. If that frequency belongs to the set $\Omega $, the associated component to that eigenvector is assigned to the trend, cycle or seasonal component and, on the contrary it is assigned to the residual $\widehat{e}_{t}$.

As for Circulant SSA we perform regressions as in (\ref{reg}) between
simulated an estimated components and check $a=0$ and $b=1.$Table \ref{Tabla_1} shows
the quantiles 10000 estimated values for $a$ and $b$. Results are very
similar for the three versions of SSA and it can be accepted that the
estimated values are close to $a=0$ and $b=1$. These simulations allow to
conclude that, at least for the proposed linear model, empirically the three versions
of SSA are equivalent. However, some differences can be found in the
estimation of the cycle, where the distribution of the estimates of $a$ \
and $b$ show less dispersion around $0$ and $1$ with CiSSA$.$

\subsection{Non-linear time series}

For the case of non-linear time series, we borrow the model from \cite{Durbin_Koopman_2012} for UK travellers given by
\begin{equation*}
x_{t}=T_{t}+c_{t}+\exp (a_{0}+a_{1}T_{t})\gamma _{t}+e _{t}
\end{equation*}%
where $T_{t}$ is the trend, $c_{t}$ is the cycle and $\gamma _{t}$ is the
seasonal component specified as in (\ref{trend}), (\ref{cycle}) and (\ref%
{seasonal}), respectively. The parameters $a_{0}$ and $a_{1}$ are unknown
fixed coefficients. Coefficient $a_{0}$ scales the seasonal component. The
sign of the coefficient $a_{1}$ determines whether the seasonal variation
increases or decreases when a positive change in the trend occurs. The
overall time varying amplitude of the seasonal component is determined by
the combination $a_{0}+a_{1}\mu _{t}.$

\begin{table}[p]
\begin{center}
\begin{spacing}{1.5}
\begin{tabular}{c | c | c c c c c}
\hline
\multirow{2}{*}{\textbf{Statistic}} & \multirow{2}{*}{\textbf{Component}} & \multicolumn{5}{c}{\textbf{Quantiles}} \\
& & 5 & 25 & 50 & 75 & 95 \\ \hline
\multicolumn{7}{c}{\textbf{Circulant SSA}} \\ \hline
\multirow{3}{*}{$\hat{a}$} & Trend & -0.0603 & -0.0199 & 0.0004 & 0.0202 & 0.0609 \\
& Cycle & -0.0111 & -0.0045 & -0.0001 & 0.0043 & 0.0112 \\
& Seasonal & -0.0015 & -0.0006 & 0.0000 & 0.0006 & 0.0015 \\ \hline
\multirow{3}{*}{$\hat{b}$} & Trend & 0.9742 & 0.9951 & 1.0037 & 1.0154 & 1.0682 \\
& Cycle & 0.8442 & 0.9567 & 1.0029 & 1.0475 & 1.1353 \\
& Seasonal & 0.9241 & 0.9779 & 1.0072 & 1.0335 & 1.0720 \\ \hline
\multicolumn{7}{c}{\textbf{Basic SSA}} \\ \hline
\multirow{3}{*}{$\hat{a}$} & Trend & -0.0602 & -0.0198 & 0.0005 & 0.0199 & 0.0605 \\
& Cycle & -0.0167 & -0.0065 & 0.0000 & 0.0066 & 0.0163 \\
& Seasonal & -0.0035 & -0.0010 & 0.0000 & 0.0009 & 0.0030 \\ \hline
\multirow{3}{*}{$\hat{b}$} & Trend & 0.9880 & 1.0064 & 1.0158 & 1.0337 & 1.1284 \\
& Cycle & 0.7626 & 0.9588 & 1.0158 & 1.0763 & 1.2660 \\
& Seasonal & 0.9269 & 0.9888 & 1.0236 & 1.0561 & 1.1084 \\ \hline
\multicolumn{7}{c}{\textbf{Toeplitz SSA}} \\ \hline
\multirow{3}{*}{$\hat{a}$} & Trend & -0.0581 & -0.0191 & 0.0002 & 0.0195 & 0.0602 \\
& Cycle & -0.0176 & -0.0063 & -0.0001 & 0.0064 & 0.0185 \\
& Seasonal & -0.0019 & -0.0007 & -0.0001 & 0.0006 & 0.0016 \\ \hline
\multirow{3}{*}{$\hat{b}$} & Trend & 0.9814 & 1.0004 & 1.0093 & 1.0284 & 1.1424 \\
& Cycle & 0.7609 & 0.9812 & 1.0513 & 1.1279 & 1.2767 \\
& Seasonal & 0.9351 & 0.9977 & 1.0315 & 1.0667 & 1.1316 \\ \hline
\end{tabular}
\end{spacing}
\caption{Statistics related to the goodness of fit of the extracted signals for the different methods. Simulations for the non-linear model, N=10000. Columns show the quantiles of the empirical distribution of the estimated coefficients of the regression of the generated components over the estimated ones.}
\label{Tabla_3}
\end{center}
\end{table}

As for the linear case, we simulate the model 10000 times for series of
length $T=193$ observations. \ We set $a_{0}$ and $a_{1}$ such that for each
replication $0.5\leq \exp (a_{0}+a_{1}\mu _{t})\leq 1.5$, with $a_{1}>0$. We
apply Circulant SSA with a window length $L=48.$ Table \ref{Tabla_3} shows the quantiles
of the empirical distribution of the estimated coefficients of the
regression in (\ref{reg}) and again we can see that the values of $a$ and $b$
estimated are located around $0$ and $1$ respectively with low dispersion.

In order to check that the estimated residuals are white noise, we fit an
AR(1) to $\widehat{e}_{t}$ as in the linear case. Table \ref{Tabla_4} shows the
quantiles of the empirical distribution of the mean, standard error and
autoregressive coefficient of the residuals of the 10000 replications. The
median of the mean and autoregressive coefficient are close to zero. The
median of the standard deviation is 0.053 (the value used for the
simulations was 0.06).

\begin{table}[h]
\begin{center}
\bigskip
\begin{spacing}{1.5}
\begin{tabular}{c | c c c c c}
\hline
\multirow{2}{*}{\textbf{Statistic}} & \multicolumn{5}{c}{\textbf{Quantiles}} \\
& 5 & 25 & 50 & 75 & 95 \\ \hline
Average & -0.0034 & -0.0011 & 0.0000 & 0.0012 & 0.0033 \\
Standard deviation & 0.0476 & 0.0508 & 0.0531 & 0.0554 & 0.0590 \\
AR(1) coefficient & -0.1727 & -0.0899 & -0.0339 & 0.0250 & 0.1066 \\ \hline
\end{tabular}
\end{spacing}
\caption{Statistics related to the residual term $\widehat{e}_{t}$ in Circulant SSA: Average, standard deviation and autoregressive coefficient of AR(1). Simulations for the non-linear model, N=10000.}
\label{Tabla_4}
\end{center}
\end{table}

As in the linear case, the results from the simulations seem very good. To
compare Circulant SSA with alternative algorithms as Basic and Toeplitz SSA,
we repeat the simulations described in the previous section and apply the
same steps to obtain their trend, cycle and seasonal components. Again we
perform regressions as in (\ref{reg}) between simulated an estimated
components and check $a=0$ and $b=1.$ Table \ref{Tabla_3} shows the quantiles 10000
estimated values for $a$ and $b$. The same conclusions as in the linear case
apply: it can be accepted that the estimated values are close to $a=0$
and $b=1$; empirically, the three versions of SSA are equivalent for the proposed linear model; and some differences can be found in the cycle estimations, where the
distribution of the estimates of $a$ and $b$ show less dispersion around $0$ and $1$ with SSA.

\section{ Application}

We consider monthly series of Industrial Production (IP), index 2010=100, of
six countries: France, Germany, Italy, UK, Japan and US. Industrial
Production is widely followed since it is pointed out in the definition of a
recession by the National Bureau of Economic Research (NBER), \url{http://www.nber.org/cycles/recessions.html}, as one of the four monthly
indicators series to check in the analysis of the business cycle. The sample
covers from January 1970 to September 2014, so the sample size $T=537.$ The
data source is the IMF database. As it can be seen in Figure \ref{Figura_1}, these
indicators show different trend, seasonality and cyclical behavior, and our
goal is to extract these components and discuss about the results.

\begin{figure}[h]
	\caption{Original IP and trend for the different countries.}
	\includegraphics[scale=0.55]{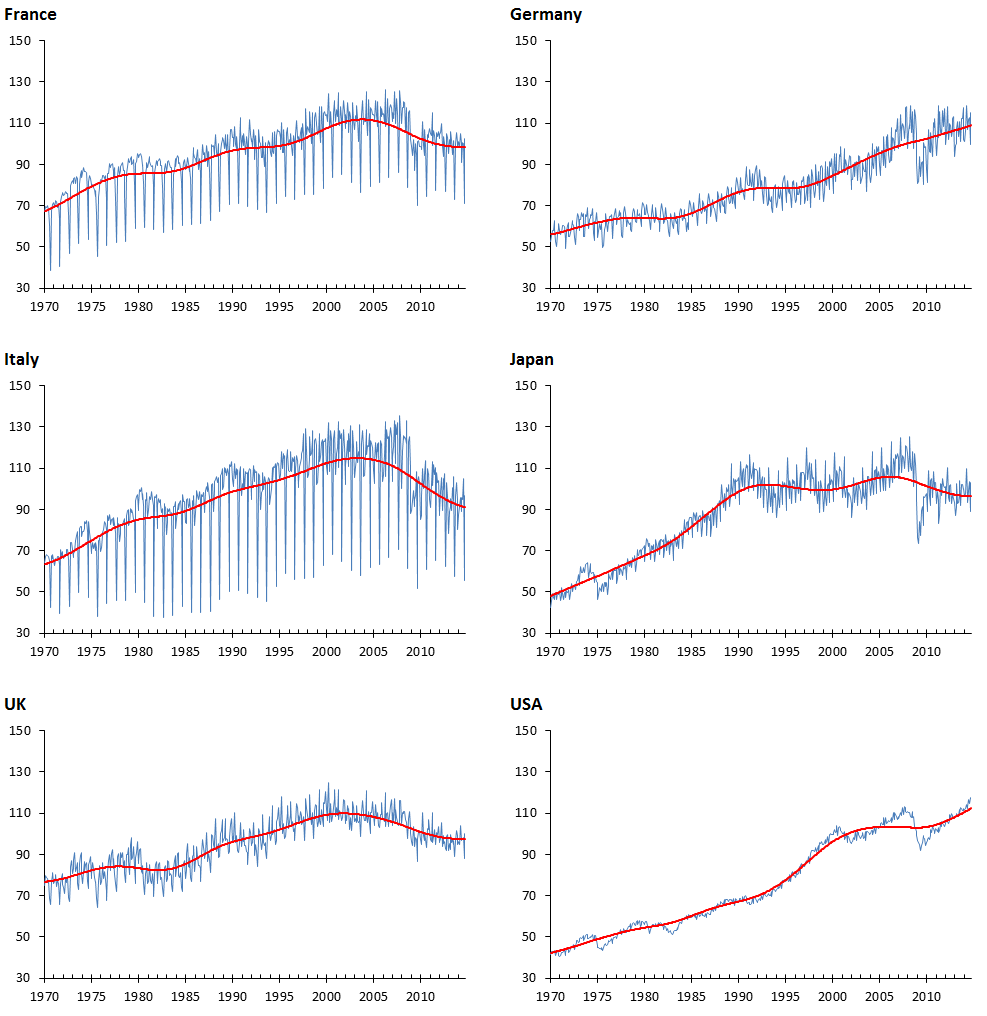}
	\centering
	\label{Figura_1}
\end{figure}

The first step is to establish the window length. Due to the monthly
periodicity and seasonality, we select a window length multiple of 12.
Assuming that the period of the cycle in these series goes from 1 year and a
half to 8 years, we choose a window length multiple of \ 8$\times $12=96
months. From the two available options, 96 and 192 months, we select the
second one since it is larger.

According to (\ref{auto})\ for $k=1,$ we have $w_k=\frac{k-1}{L}=0$ and it
will be associated to the trend. In the same way, for $k=2$, we have $w=1/192
$\thinspace , that corresponds to 192 months or 16 years that are beyond
cyclical movements between 1.5 and 8 years. Therefore, given (\ref{auto})\
and the symmetry of the spectral density, the trend is reconstructed with
the eigentriples 1, 2 and $L+2-k=192$ with the elementary groups by
frequencies from $B_{1}$ and $B_{2}$ respectively. In an analogous way,
assuming that cycle goes from 1.5 to 8 years, this component is associated
to the frequencies $w_k=1/96,1/64,1/48,5/192,1/32,7/192,1/24,3/64,5/96$ and
the cycle signal is reconstructed with the eigentriples 3 to 11 and 183 to
191, with the elementary groups by frequencies from $B_{3}$ to $B_{11}.$
Finally, the seasonal component is associated to the frequencies $%
w_k=1/12,1/6,1/4,1/3,5/12,1/2$ and reconstructed in a similar way with the
eigentriples 17, 33, 49, 65, 81, 97, 113, 129, 145, 161 and 177 and with the
elementary groups by frequencies $B_{17},B_{33},B_{49},B_{65},B_{81},$ and $%
B_{97}.$

Table \ref{Tabla_5} shows the contributions of the signals to the original IP variations
in percentage. First, we highlight that the contribution of the irregular
component (those oscillations not explained by the trend, cyclical or
seasonal components) is smaller than 3.5\% in all the countries. Main
contributions come from the trend and seasonality, that account for more
than 84\% in all the countries. As expected, the contribution of the
seasonal component is almost negligible in US, and quite small in Japan and
Germany, while it is very relevant in Italy and France. Finally, the cycle
contributes in a range between 7.8\% in Italy to 13.8\% in Japan.

\begin{table}[h]
\begin{center}
\bigskip
\begin{spacing}{1.5}
\begin{tabular}{c | c c c c c c}
\hline
\multirow{2}{*}{\textbf{Component}} & \multicolumn{6}{c}{\textbf{Country}} \\
& France & Germany & Italy & Japan & UK & USA \\ \hline
Trend & 52.1 & 77.3 & 42.7 & 79.0 & 72.0 & 87.9 \\
Cycle & 9.5 & 12.6 & 7.8 & 13.8 & 11.1 & 10.3 \\
Seasonal & 35.6 & 6.7 & 47.3 & 5.1 & 13.5 & 0.3 \\
Irregular & 2.8 & 3.4 & 2.2 & 2.1 & 3.4 & 1.5 \\ \hline
\end{tabular}
\end{spacing}
\caption{Contribution of the different signals to IP in the six countries in percentage.}
\label{Tabla_5}
\end{center}
\end{table}

Figure \ref{Figura_1} shows the estimated trends for every country. The trend is a smooth
component that has shown a decreasing evolution since the last decade for
France, Italy and UK as a consequence of the last economic crisis. On the
contrary, in Germnay and US, the trend shows an upward evolution in all the
sample period.

Figure \ref{Figura_2} shows the cyclical component where the shaded areas correspond to
recessions as dated by the OECD. We can see that the extracted cycle
reflects quite well the business cycle for all countries.

\begin{figure}[h]
	\caption{Estimated cycles and OECD announced recessions (shadowed areas).}
	\includegraphics[scale=0.55]{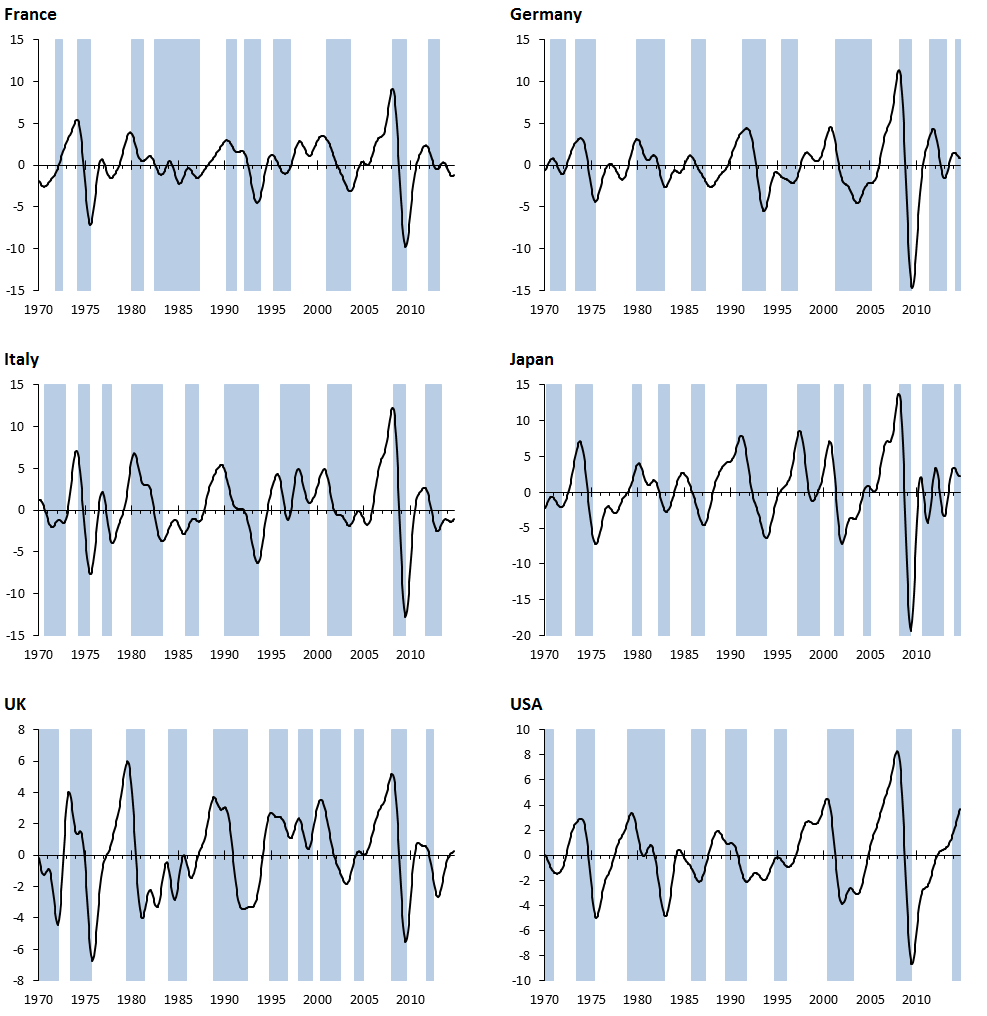}
	\centering
	\label{Figura_2}
\end{figure}

\subsection{Separability of the estimated components with CiSSA}

One desirable property of the signal extraction method is that the resulting
components should be orthogonal. However, in practice, they usually exhibit
cross-correlation. Residual seasonality in seasonal adjusted time series is
another concern in any signal extraction method from very early times \cite{Burman_1980, Dagum_1978}, and it is still a matter of interest nowadays. Findley et al. \cite{Findley_others_2017} point out that "The most fundamental seasonal adjustment deficiency is detectable seasonality after
adjustment". This is also a concern for policy makers \cite{Moulton_Cowan_2016}.

Separability of the elementary series as well as those grouped by
frequencies is an assumption of SSA and should also be a characteristic of
the estimated components. This characteristic is important since many signal
extraction procedures assume zero correlation between their underlying
components, whereas the estimated signals can be quite correlated. The SSA decomposition can be successful only if the resulting additive components of the series are quite separable
from each other \cite{Golyandina_others_2001}.

For a fixed window length $L,$ given two series $\left\{ x_{t}^{(1)}\right\} 
$ and $\left\{ x_{t}^{(2)}\right\} $ extracted from the series $\left\{
x_{t}\right\} $, we say that they are weakly separable if both their column
as well as row spaces are orthogonal, that is $\mathbf{X}^{(1)}\left( 
\mathbf{X}^{(2)}\right) ^{^{\prime }}=\mathbf{0}_{L\times L}$ and $\left( 
\mathbf{X}^{(1)}\right) ^{^{\prime }}\mathbf{X}^{(2)}=\mathbf{0}_{N\times N}$%
. Furthermore, we say that two series $\left\{ x_{t}^{(1)}\right\} $ and $%
\left\{ x_{t}^{(2)}\right\} $ are strongly separable if they are weakly
separable and the two sets of singular values of the trajectory matrices $%
\mathbf{X}^{(1)}$ and $\mathbf{X}^{(2)}$ are disjoint. When the trajectory
matrix of the original time series has not multiple singular values or,
equivalently, each elementary reconstructed series belongs to a different
harmonic, strong separability is guaranteed according to the previous
definition.

Usually, separability is measured in terms of \textbf{w}-correlation \cite{Golyandina_others_2001, Golyandina_Zhigljavsky_2013}, that it is given by 
\begin{equation*}
\rho _{_{12}}^{w}=\frac{\left\langle \mathbf{x}^{(1)},\mathbf{x}%
^{(2)}\right\rangle _{w}}{\left\Vert \mathbf{x}^{(1)}\right\Vert
_{w}\left\Vert \mathbf{x}^{(2)}\right\Vert _{w}}\:,
\end{equation*}%
where $\left\langle \mathbf{x}^{(1)},\mathbf{x}^{(2)}\right\rangle _{w}=(%
\mathbf{x}^{(1)})^{\prime }\mathbf{Wx}^{(2)}\,$\ is the so called w-inner
product, $\left\Vert \mathbf{x}^{(1)}\right\Vert _{w}=\sqrt{\left\langle 
\mathbf{x}^{(1)},\mathbf{x}^{(1)}\right\rangle _{w}}$ and $\mathbf{W}%
=\operatorname{diag}\left( 1,2,\cdots , \right.\underbrace{L,\cdots ,L}_%
{T-2\left( L-1 \right)\text{ times}}\left. ,\cdots ,2,1 \right)$.
Note that the window length $L$ enters the definition of 
\textbf{w}-correlation. We are interested on producing components with 
\textbf{w}-correlation (ideally) zero because, in this case, we can conclude
that the component series are \textbf{w}-orthogonal, i. e. $\left\langle 
\mathbf{x}^{(1)},\mathbf{x}^{(2)}\right\rangle _{w}=0$ and separable \cite{Golyandina_others_2001}.

To show that Circulant SSA produces components that are strongly separable,
first notice that the real eigenvectors $\sqrt{2}R_{\mathbf{u}_{k}}$ and $%
\sqrt{2}I_{\mathbf{u}_{k}}$ (linked to eigenvalues $\lambda _{k}$ and $%
\lambda _{L+2-k}$, respectively, $\lambda _{k}=\lambda _{L+2-k}$) are
orthogonal and have information associated only to frequency $\frac{k-1}{L}$
. Those are the only eigenvectors that have information related to this
frequency. As eigenvectors can be considered filters \cite{Kume_2013, Tome_others_2018}, these
pair of eigenvectors extract elementary series linked to the same frequency
without mixing harmonics of other frequencies. As a result, the two
elementary series, when reconstructed in step 4, have spectral correlation
close to 1 between them and close to zero with the remaining ones. Taking
into account the pairs of reconstructed series per frequency, any grouping
of the reconstructed series results in disjoint sets from the point of view
of the frequency. Then, Circulant SSA produces components that are
approximately strongly separable.

\begin{figure}[h]
	\caption{w-correlation matrix for the elementary reconstructed series for the 30 greatest eigenvalues.}
	\includegraphics[scale=0.55]{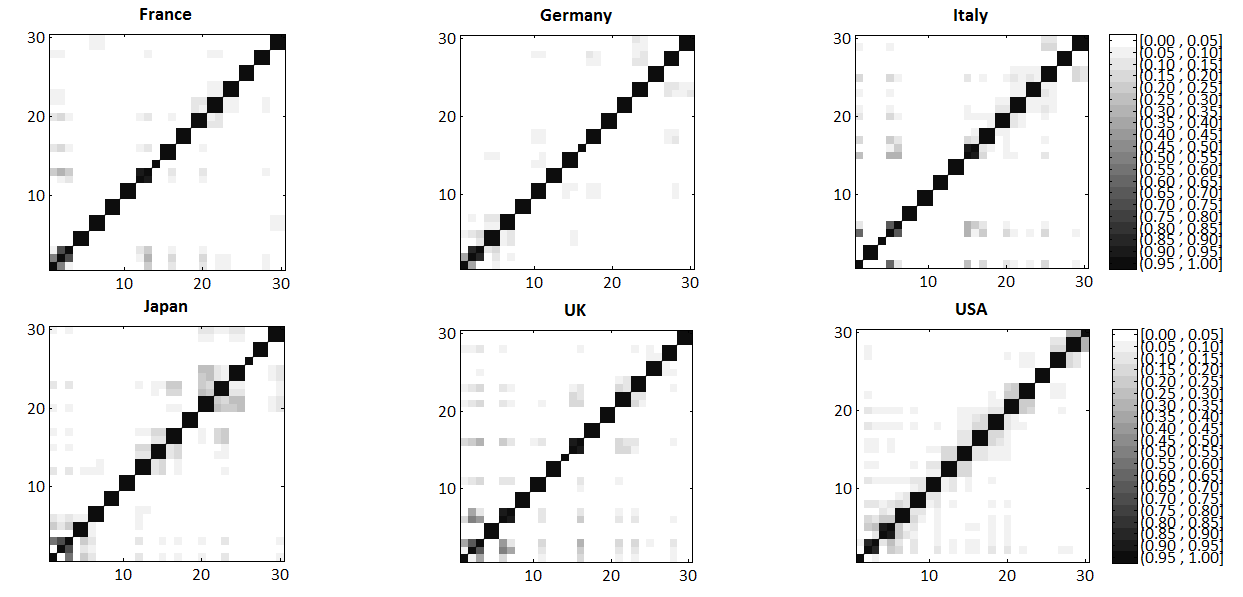}
	\centering
	\label{Figura_3}
\end{figure}

To quickly check how separable the components are, Figure \ref{Figura_3} plots the matrix
of the absolute values of the \textbf{w}-correlations for all the IP
components, coloring in white the absence of \textbf{w}-correlation, in
black \textbf{w}-correlations in absolute value equal to 1 and in a scale of
grey colors the remaining intermediate values. It can be seen that, as
expected, Circulant SSA produces components that are strongly separable.

Furthermore, seasonal adjusted time series for Industrial Production are
largely followed by real time analysts, and one desirable property is that
they have no remaining seasonality. To check the quality of seasonal
adjustment by Circulant SSA, we have applied the \textit{combined test for
seasonality} \cite{Lothian_1978} used in X12-ARIMA. We found that there were no
signs of any remaining seasonality in any of the seasonal adjusted time
series for the different countries \footnote{Results are available from the authors upon request.}.

\section{Conclusions}

In this paper we propose CiSSA,
Circulant SSA, an automated procedure that allows to extract the signal associated to any given frequency specified beforehand. This is a different to previous versions of SSA that, after extracting the principal components of the trajectory matrix, they need to identify their frequency of oscillation and group them in order to form the desired signals. 

CiSSA relies on the eigenstructure of a circulant matrix related to the second moments of the time series. Circulant matrices have closed form solutions for their eigenvalues and eigenvectors. Additionally, we can use them to evaluate the spectral density at specific frequencies. 
We prove that CiSSA is asymptotically equivalent to Basic and Toeplitz SSA.

We also extend the algorithm of Circulant SSA to the nonstationary case providing a
generalization of Gray's theorem.

The properties of Circulant SSA have been checked through a set of
simulations for linear and nonlinear time series models as well as through
the empirical application where we showed that Circulant SSA produces
deseasonalized series clean of any seasonality. The estimated cycles also
matched the business cycles dating proposed by the OECD.

\section{Appendix}

The proof of Theorem 1 relies on a set of lemmas and propositions that we need to shown before. Proposition 1 shows the asymptotic equivalence between the Toeplitz matrices of sample and population second moments, $\mathbf{S}_T \sim \mathbf{\Gamma}_L(f)$. Proposition 2 shows that the sequence of matrices $\mathbf{S}_B$ are also asymptotically equivalent to the Toeplitz matrix of population second moments $\mathbf{\Gamma}_L(f)$. We also need two auxilliary lemmas regarding probability convergence of sample and population second moments. 

\begin{lemma}
For a stationary time series, the sequence $S_{L}=\sum_{m=0}^{L-1}\left( 
\widehat{\gamma }_{m}-\gamma _{m}\right) ^{2}$ converges in probability to 0
when $L\longrightarrow \infty $.
\end{lemma}

\begin{proof}
The sum $S_{L}$ can be decomposed as
\[{{S}_{L}}=\sum\limits_{m=0}^{L-1}{{{\left({{{\hat{\gamma }}}_{m}}-{{\gamma }_{m}}\right)}^{2}}}=\sum\limits_{m=0}^{L-1}{\gamma_{m}^{2}}+\sum\limits_{m=0}^{L-1}{\hat{\gamma}_{m}^{2}}-2\sum\limits_{m=0}^{L-1}{\gamma _{m}^{{}}\hat{\gamma }_{m}^{{}}}\:.\]
The first term in the previous equation is finite when $L\longrightarrow \infty $ by
Parseval's Theorenm, that is $\sum_{m=0}^{\infty }\gamma
_{m}^{2}=K$. Preserving $L<T/2$, $L$ is a monotonically increasing sequence as a function
of $T$ so $L\longrightarrow \infty $ when $T\longrightarrow \infty $. Thus if $L\longrightarrow \infty $ means that $T\longrightarrow \infty $ and, therefore, the sum of infinite addends of the second term converges in probability to $K$ given that $\widehat{\gamma }_{m} \longrightarrow \gamma _{m}$ in probability when $T\longrightarrow \infty$. And, because of the same reasoning, the third
term converges in probability to $2K$ when $L\longrightarrow \infty $. As a
consequence, the sum $S_{L}$ converges in probability to $K+K-2K=0$ as $L\longrightarrow \infty $.
\end{proof}

\begin{proposition}
Let $\left\{\mathbf{ S}_{T}\right\} $ and $\left\{ \mathbf{\Gamma }_{L}\left( f\right)
\right\} $ be two sequences of matrices defined in function of the window
length defined by \ref{vg} and \ref{Gamma_L} respectively. Then, $\mathbf{ S}_{T}\sim \mathbf{\Gamma }_{L}\left(
f\right) $
\end{proposition}

\begin{proof}
We know that the eigenvalues of the Toeplitz matrix $\mathbf{\Gamma }_{L}\left(
f\right) $ are bounded \cite{Tilli_1998}. The matrices $\mathbf{ S}_{T}$ are Toeplitz and symmetric,
therefore their real eigenvalues are also bounded. We must proof that $%
\underset{L\longrightarrow \infty }{\lim }\frac{1}{L}\left\Vert\mathbf{ S}_{T}-\mathbf{\Gamma}
_{L}\left( f\right) \right\Vert _{F}=0$. We can write
\begin{eqnarray*}
0 & \leq &\frac{1}{L}\left\Vert\mathbf{ S}_{T}-\mathbf{\Gamma }_{L}\left( f\right) \right\Vert
_{F}^{2}=\frac{1}{L}\sum_{m=1-L}^{L-1}\left( L-m\right) \left( \widehat{%
\gamma }_{m}-\gamma _{m}\right) ^{2} \leq \\
& \leq & 2\sum_{m=0}^{L-1}\frac{L-m}{L}%
\left( \widehat{\gamma }_{m}-\gamma _{m}\right) ^{2} \leq \\
&\leq&
2 \sum_{m=0}^{L-1}\left( \widehat{\gamma }_{m}-\gamma _{m}\right) ^{2}.
\end{eqnarray*}
By the \emph{Squeeze Theorem} and the previous Lemma, we obtain that $%
\underset{L\longrightarrow \infty }{\lim }\frac{1}{L}\left\Vert\mathbf{ S}_{T}-\mathbf{\Gamma}
_{L}\left( f\right) \right\Vert _{F}=0$ and therefore it is proved that $%
\mathbf{S}_{T}\sim \mathbf{\Gamma }_{L}\left( f\right)$ .
\end{proof}

In Basic SSA, it is possible to substitute the matrix $\mathbf{S}=\mathbf{XX}'$ by $\mathbf{S}_{B}=\mathbf{XX}'/N$ for stationary time series \cite{Golyandina_others_2001}.
Matrices $\mathbf{S}$ and $\mathbf{S}_{B}$, with dimension $L\times L$, have the same
eigenvalues and the eigenvectors of $\mathbf{S}_{B}$ are those of $\mathbf{S}$
multiplied by $1/N$. The elements of matrix $\mathbf{S}_{B}$ are given by $\widetilde{s}_{ij}=\frac{1}{N}\sum_{t=1}^{i+N-1}x_{t}x_{t+j-i}$ and, under stationarity, it holds that $\widetilde{s}_{ij}$ converges to $
\gamma _{\left\vert t-j\right\vert \text{ }}$as $N\longrightarrow \infty $,
it is, when $T\longrightarrow \infty $. From matrix $\mathbf{S}_{B}$ we
obtain a sequence of symmetric matrices $\left\{ \mathbf{S}_{B}\right\} $ as
a function on the window lenght $L$. To relate this sequence $\left\{ 
\mathbf{S}_{B}\right\} $ of matrices symmetric with the sequence of Toeplitz
symmetric matrices $\left\{ \mathbf{\Gamma }_{L}\left( f\right) \right\} $ we must
proof the following Lemma.

\begin{spacing}{1}
\begin{lemma}
Under stationarity, the sequence $S_{L}=\sum_{m=0}^{L-1}\underset{
\begin{array}{c}1\leq i,j\leq L \\ \left\vert i-j\right\vert =m
\end{array}}{\max }\left\{ \left( \widetilde{s}_{ij}-\gamma _{m}\right) ^{2}\right\}$ converges in probability to $0$ when $L\longrightarrow \infty $.
\end{lemma}
\end{spacing}

\begin{proof}
The sum $S_{L}$ verifies that
\begin{spacing}{1}
\begin{equation*}
0\leq S_{L} \leq \sum_{m=0}^{L-1}\gamma _{m}^{2}+\sum_{m=0}^{L-1}\underset{%
\begin{array}{c}
1\leq i,j\leq L \\ 
\left\vert i-j\right\vert =m%
\end{array}%
}{\max }\left\{ \widetilde{s}_{ij}^{2}\right\} -2\sum_{m=0}^{L-1} \underset{%
\begin{array}{c}
1\leq i,j\leq L \\ 
\left\vert i-j\right\vert =m%
\end{array}%
}{\min }\left\{ \gamma _{m}\widetilde{s}_{ij}\right\} .
\end{equation*}
\end{spacing}

\bigskip
By Parseval's Theorem, the first term on the right is finite as $L\longrightarrow \infty $, it is quadratic summable, $\sum_{m=0}^{L-1}\gamma_{m}^{2}=K$. We know that $N=T-L+1$. Given that $L<T/2$, $N>T/2+1$ and, further $L$ and $N$ are monotonically increasing sequences as functions of $T$, so $L, N\longrightarrow \infty $, as $T\longrightarrow \infty $. Therefore, if $L\longrightarrow \infty $ means that $T\longrightarrow \infty $ and the sum of infinite addends of the second term converges in probability to $K$ because $\widetilde{s}_{ij}\longrightarrow \gamma
_{\left\vert i-j\right\vert },$for all $i,j$, when $N\longrightarrow \infty $, that is, when $T\longrightarrow \infty$.
And, following the same reasoning, the third term converges in probability
to $2K$ when $L\longrightarrow \infty $. Therefore the right term of the
inequality converges to $0$ in probability. Finally by the Squeze
Theorem, $S_{L}$ converges in probability to $0$.
\end{proof}

\begin{proposition}
Let $\left\{ \mathbf{S}_{B}\right\}$ and $\left\{ \mathbf{\Gamma }_{L}\left( f\right)
\right\}$ be the sequences of matrices defined as a function of the window
length $L$. Then $\mathbf{S}_{B}\sim \mathbf{\Gamma }_{L}\left( f\right) $.
\end{proposition}

\begin{proof}
The eigenvalues of the Toeplitz matrix $\mathbf{\Gamma }_{L}\left(f\right)$ are bounded \cite{Tilli_1998}. The symmetric matrices $\mathbf{S}_{B}$ converge to Toeplitz matrix in probability. Then, their eigenvalues are bounded in probability. Now we must proof that $\underset{L\longrightarrow \infty }{\lim }
\frac{1}{\sqrt{L}}\left\Vert \mathbf{S}_{B}-\mathbf{\Gamma }_{L}\left( f\right)
\right\Vert _{F}=0$. We can write,
\begin{spacing}{1}
\begin{eqnarray*}
0 & \leq & \frac{1}{L}\left\Vert \mathbf{S}_{B}-\mathbf{\Gamma }_{L}\left( f\right)
\right\Vert _{F}^{2}=\frac{1}{L}\sum_{i=1}^{L}\sum_{j=1}^{L}\left( 
\widetilde{s}_{ij}-\gamma _{\left\vert i-j\right\vert }\right) ^{2} \leq \\
& \leq&
2\sum_{m=0}^{L-1}\frac{L-m}{L}\underset{%
\begin{array}{c}
1\leq i,j\leq L \\ 
\left\vert i-j\right\vert =m%
\end{array}%
}{\max }\left\{ \left( \widetilde{s}_{ij}-\gamma _{m}\right) ^{2}\right\} \leq \\
& \leq & 2\sum_{m=0}^{L-1}\underset{%
\begin{array}{c}
1\leq i,j\leq L \\ 
\left\vert i-j\right\vert =m%
\end{array}%
}{\max }\left\{ \left( \widetilde{s}_{ij}-\gamma _{m}\right) ^{2}\right\} .
\end{eqnarray*}
\end{spacing}
Therefore, by the Squeeze Theorem and previous Lemma, it holds that $\underset{L\longrightarrow \infty }{\lim }\frac{1}{L}\left\Vert \mathbf{S}_{B}-\mathbf{\Gamma }_{L}\left( f\right) \right\Vert _{F}^{2}=0$ and $\mathbf{S}_{B}\sim\mathbf{\Gamma }_{L}\left( f\right) $.
\end{proof}

\bigskip
\textbf{Proof of Theorem 1:} 
We have that $\mathbf{S}_{T}\sim \mathbf{\Gamma }_{L}\left( f\right) $ and $\mathbf{S}_B\sim
\mathbf{\Gamma }_{L}\left( f\right) $ by propositions 2 and 4 respectively, and that together with the transitive property lead to $\mathbf{S}_B\sim\mathbf{ S}_{T}$. Given that by construction $\mathbf{S}_{T}\sim \mathbf{S}_{C}$ \cite{Pearl_1973} and, again, by transitive property we have that $\mathbf{S}_B\sim \mathbf{S}_{C}$. \qquad $\blacksquare$

\bigskip
\textbf{Proof of Theorem 2:} As defined, the function $g(w)$ is real, continuous and 2$\pi $-periodic. Its image is $\left[ \frac{1}{M_{s}} ,\max_{i}M_{h_{i}}\right] $ being different from zero in the whole interval.
Then, by the properties of the inverse of Toeplitz matrices $\left( \mathbf{T }_{L}(g^{-1})\right) ^{-1}\sim \mathbf{T}_{L}(g)$. Moreover, if $F(x)$ is continuos in $\left[ \frac{1}{M_{s}},\max_{i}M_{h_{i}}\right] $, then $F(\frac{1}{x})$ is continuos in $\left[ \frac{1}{\max_{i}M_{h_{i}}},M_{s}\right]$. Since the assumption of $g(w)$ being a Wiener's class function relaxes to a continuous and 2$\pi$-periodic function \cite{Gutierrez_Crespo_2012}, Szeg\"{o}'s theorem leads to (\ref{teorema}).\qquad $\blacksquare$

\singlespacing

\bigskip

\end{document}